\documentclass[pra,twocolumn,showpacs,superscriptaddress]{revtex4}

\usepackage{amsfonts,amsmath,amssymb,subfigure,graphicx,bm,times,color,amsthm} 

\newcommand{\suma}[1]{\sum_{{#1} \in \mathbb{Z}}}
\newcommand{\intcirc}{\int_{2\pi}}
\newcommand{\op}[1]{\hat{#1}}
\newcommand{\Tr}{\mathop{\mathrm{Tr}}\nolimits}

\newtheorem{lemma}{Lemma}
\newtheorem*{theorem}{Theorem}

\begin{document}

\title{Non-negative Wigner functions for orbital angular momentum states}

\author{I.~Rigas} 
\affiliation{Departamento de \'Optica, 
Facultad de F\'{\i}sica, Universidad Complutense, 
28040~Madrid, Spain}

\author{L.~L.~S\'{a}nchez-Soto} 
\affiliation{Departamento de \'Optica,
Facultad de F\'{\i}sica, Universidad Complutense, 
28040~Madrid, Spain}

\author{A.~B.~Klimov} 
\affiliation{Departamento de F\'{\i}sica,
Universidad de Guadalajara, 
44420~Guadalajara, Jalisco, Mexico}

\author{J.~\v{R}eh\'{a}\v{c}ek} 
\affiliation{Department of Optics,
Palacky University, 17. listopadu 50, 
772 00 Olomouc, Czech Republic}

\author{Z.~Hradil} 
\affiliation{Department of Optics, 
Palacky University, 17. listopadu 50, 
772 00 Olomouc, Czech Republic}

\date{\today}

\begin{abstract}
  The Wigner function of a pure continuous-variable quantum state is
  non-negative if and only if the state is Gaussian. Here we show that
  for the canonical pair angle and angular momentum, the only pure states
  with non-negative Wigner functions are the eigenstates of the angular
  momentum. Some implications of this surprising result are discussed.
  \end{abstract}

\pacs{03.65.Fd,03.65.Ta,03.65.Sq,03.67.-a}

\maketitle

For continuous variables, the Wigner function~\cite{Wigner:1932} is a
very useful tool that establishes a one-to-one correspondence between
quantum states and joint quasiprobability distributions of canonically
conjugate variables in phase space (position and momentum, in the
standard case).  However, it can take on negative values, a 
property that distinguishes it from a true probability
distribution~\cite{Hillery:1984,Lee:1995,QMPS:2005}.  Indeed, 
this negative character is associated with the existence of quantum 
interference, which itself may be identified as a signal of 
nonclassical behavior~\cite{Kenfack:2004lw}.

In consequence, the characterization of quantum states that are
classical, in the sense of giving rise to non-negative Wigner  
functions, is a topic of undoubted interest. Among pure states, 
it was proven in a classical paper by Hudson~\cite{Hudson:1974kb} 
(later generalized by Soto and Claverie~\cite{Soto:1983dx} to
multipartite systems) that the only states that have non-negative
Wigner functions are Gaussian
states~\cite{Janssen:1984kx,Lieb:1990yq}.  
This is one of the main reasons for the prominent role these states 
play in modern quantum information~\cite{Cerf:2007qg}.

The original definition of the Wigner function has also been extended
to discrete systems (see Ref.~\cite{Bjork:2008ab} for a comprehensive
review).  Again, the classification of states with non-negative Wigner
functions is an amazing problem that has been solved quite recently by
Paz and coworkers~\cite{Cormick:2006ty,Cormick:2006jj} and
Gross~\cite{Gross:2006py,Gross:2007nt}, so that the role of Gaussian
states is now taken on by stabilizer states.  Interestingly, these
are the only states that can be simulated efficiently in classical
computers~\cite{Gottesman:1997pd}.

Between these two cases (whose proofs are otherwise completely
different), we have the interesting situation of canonical pairs, such
as the angle and orbital angular momentum (OAM), for which one variable is
continuous while the other one is discrete~\cite{Kastrup:2006wh}.  The
associated phase space is the discrete cylinder $\mathcal{S}_1 \times
\mathbb{Z}$, where $\mathcal{S}_{1}$ stands for the unit circle
(associated to the angle) and the integers $\mathbb{Z}$ translate the
discreteness of the OAM. The physical example we have in mind is the
OAM of photons. This is an emerging field that has given rise to many
developments, ranging from optical tweezers to high-dimensional
quantum entanglement, or fundamental processes in Bose-Einstein
condensates, to cite only a few relevant examples~\cite{Allen:2003}.

The seminal paper of Allen \emph{et al.}~\cite{Allen:1992} firmly
established that the Laguerre-Gauss modes carry a well-defined
OAM. They appear as annular rings with a zero on-axis intensity and an
azimuthal dependence $\exp(i \ell \phi)$ that gives rise to spiral
wave fronts. The index $\ell$ takes only integer values and can be
seen as the eigenvalue of the OAM operator. Since then, several
methods have been established to produce light beams with the required
azimuthal phase structure, among these spiral phase plates, forked
holograms, and spatial light modulators are perhaps the most
versatile. In this way, a variety of modes with helical phase fronts
but different transverse patterns (such as Bessel, Mathieu, or
hypergeometric beams) can be routinely generated in the 
laboratory~\cite{Franke-Arnold:2008sw}.

The goal of this work is precisely to determine the pure states of
these OAM-carrying systems for which the Wigner function is
non-negative, filling in this way a long overdue gap.

To be as self-contained as possible, we first introduce some basic
notions for the problem at hand of cylindrical symmetry. We are
concerned with the planar rotations by an angle $\phi$ generated by
the angular momentum along the $z$ axis, which for simplicity will be
denoted henceforth as $\op{L}$.  We do not want to enter in a long
discussion about the possible existence of an angle
operator~\cite{Rehacek:2008ss}.  For our purposes here, the simplest
solution is to adopt two periodic angular coordinates, e.g.,
cosine and sine, that we shall denote by $\op{C}$ and $\op{S}$ to 
make no further assumptions about the angle itself.  One can concisely
condense all this information using the complex exponential of the
angle $\op{E} = \op{C} + i \op{S}$, which satisfies the commutation
relation
\begin{equation}
  \label{ELE} 
  [ \op{E},  \op{L} ] = \op{E} \, .
\end{equation}
In mathematical terms, this defines the Lie algebra of the
two-dimensional Euclidean group E(2), which is precisely the 
canonical symmetry group for the cylinder.

The action of $\op{E}$ on the  basis of eigenstates of $\op{L}$ is 
$ \op{E} | \ell \rangle = | \ell - 1 \rangle$, and it possesses then 
a simple implementation by means of a phase mask removing a charge $+ 1$ 
from a vortex state~\cite{Mair:2001,Hradil:2006}.  Since the integer 
$\ell$ runs from $ - \infty$ to $+ \infty$, $\op{E}$ is a unitary operator 
whose  eigenvectors
\begin{equation}
  \label{phi_states} 
  | \phi \rangle = \frac{1}{\sqrt{2 \pi}}
  \sum_{\ell \in \mathbb{Z}} e^{i \ell \phi} | \ell \rangle 
\end{equation}
form a complete basis and describe states with well-defined angle.  
In the representation generated by them, $\op{L}$ acts as
$-i \partial_{\phi}$ (in units of $\hbar = 1$). 

Given the key role played by the displacement operators in
settling the Wigner function for the harmonic oscillator, we
introduce a unitary displacement operator
\begin{equation}
  \label{eq:Displace1}
  \op{D} (\ell, \phi) = e^{i \alpha (\ell,\phi)} \, 
  \op {E}^{-\ell} e^{-i\phi \op{L}} \, ,
\end{equation}
where $ \alpha (\ell, \phi)$ is a phase required to avoid plugging
in extra factors when acting with $\op{D}$. The conditions of unitarity
and periodicity restrict the possible values of $\alpha$, although a
sensible choice is $\alpha (\ell, \phi) = - \ell \phi/2$.  Note that
here we cannot rewrite Eq.~(\ref{eq:Displace1}) as an entangled
exponential, since the action of the operator to be exponentiated
would not be well defined.

We use as a guide the analogy with the continuous case and introduce
the mapping~\cite{Berezin:1975}
\begin{equation}
  W_{\op{\varrho}} (\ell, \phi) = 
  \Tr [ \op{\varrho} \,\op{w} ( \ell,\phi) ] \, ,  
  \label{eq:WigFunDef1}  
\end{equation}
which maps the density operator into a Wigner function via a kernel 
$\op{w}$ defined as a double Fourier transform of the displacement 
operator~\cite{Plebanski:2000}:
\begin{equation}
  \label{eq:WigKerDef1}
  \op{w} (\ell, \phi) = 
  \frac{1}{(2\pi)^2} 
  \suma{\ell^\prime} \intcirc   
  \exp[-i ( \ell^\prime \phi - \ell \phi^\prime)] \,  
  \op{D} (\ell^\prime, \phi^\prime) \, d\phi^\prime \, ,
\end{equation}
where the integral extends to the $2\pi$ interval within which the
angle is defined. This mapping is invertible, so one can reconstruct the 
density operator as
\begin{equation}
  \label{eq:1}
  \op{\varrho} = 2 \pi \, 
 \suma{\ell} \intcirc \op{w} (\ell,\phi) \, W_{\op{\varrho}} (\ell,\phi) \, d\phi \, .
\end{equation}

The (Hermitian) Wigner kernels $ \op{w} (\ell, \phi)$ are a complete
orthonormal basis (in the trace sense) for the operators acting on the
Hilbert space of the system. In addition, they are explicitly
covariant; i.e., they transform properly under displacements, $ \op{w}
(\ell,\phi) = \op{D}(\ell,\phi) \, \op{w}(0,0) \,\op{D}^\dagger
(\ell,\phi)$.  In fact, these properties guarantee that the Wigner
function defined in Eq.(\ref{eq:WigFunDef1}) bears all the good
properties required for a probabilistic description. In particular, it
reproduces the proper marginal distributions, that is,
\begin{equation}
  \label{eq:ELMargin}
  \suma{\ell} W_{\op{\varrho}} (\ell,\phi) = 
  \langle \phi| \op{\varrho} | \phi \rangle \, ,
  \quad
  \intcirc W_{\op{\varrho}} (\ell,\phi) \,  d\phi = 
  \langle \ell| \op{\varrho} | \ell \rangle \, . 
\end{equation}
Finally, the overlap of two density operators is proportional to
the integral of the associated Wigner functions:
\begin{equation}
  \label{eq:HWProps4}
  \Tr ( \op{\varrho} \,\op{\sigma} ) \propto
  \suma{\ell}  \intcirc 
  W_{\op{\varrho}}(\ell, \phi)  W_{\op{\sigma}} (\ell, \phi) \, d\phi \, .
\end{equation}
This property (often called traciality) offers practical advantages, 
since it allows one to predict the statistics of any outcome, once the 
Wigner function of the measured state is known.

We remark that this approach to the Wigner function is grounded in the
axiomatic method developed by Stratonovich~\cite{Stratonovich:1956}
and Berezin~\cite{Berezin:1975} (see also Ref.~\cite{Brif:1998pj}). It
is possible to follow alternative routes, such as, introducing a
Wigner function as the Fourier transform of some generalized
characteristic function~\cite{Wolf:1996xr}. This has been pursued also
for the group E(2)~\cite{Nieto:1998cr}. However, these
apparently disjoint formulations turn out to be equivalent for most
practical purposes~\cite{Chumakov:2000pl}.

To give an explicit form of the Wigner function \eqref{eq:WigFunDef1}
we need to evaluate it in a basis. Using the OAM eigenstates, we get
\begin{eqnarray}
  \label{eq:ExplictKernel3} 
  W_{\op{\varrho}} (\ell,\phi)  & =  &
  \frac{1}{2\pi} \suma{\ell^\prime}  
  e^{-2 i \ell^\prime \phi} \langle \ell - \ell^\prime | \op{\varrho} 
  |\ell + \ell^\prime \rangle  \nonumber  \\
  & +  & \frac{1}{2\pi^2} \suma{\ell^{\prime},\ell^{\prime \prime}} 
  \frac{(-1)^{\ell^{\prime \prime}}}{\ell^{\prime \prime} + 1/2} 
  e^{-(2 \ell^{\prime} + 1) i \phi} \nonumber \\
  & \times &
  \langle \ell + \ell^{\prime \prime} - \ell^{\prime} | \op{\varrho}
  |\ell + \ell^{\prime \prime} + \ell^{\prime} + 1 \rangle \, .
\end{eqnarray}
This looks rather cumbersome due to the second  sum in
Eq.~(\ref{eq:ExplictKernel3}) and sometimes is preferable to work in
the angle representation, for which one easily finds
\begin{equation}
  \label{eq:WignerAngle}
  W_{\op{\varrho}} (\ell,\phi) = \frac{1}{2\pi} \int_{-\pi}^\pi \!\! 
  \langle \phi- \phi^{\prime}/2 | \op{\varrho} | \phi + \phi^{\prime}/2 \rangle 
  \, e^{i \phi^{\prime} \ell} \, d\phi^{\prime} \, .
\end{equation}
This coincides with the result of Mukunda~\cite{Mukunda:1979,Mukunda:2005} (see also Ref.~\cite{Bizarro:1994}) and bears a resemblance with the standard
Wigner function for position and momentum that is more than evident. Note that
using this latter function in terms of transverse coordinates, as is
often done in classical optics~\cite{Simon:2000jt}, is not appropriate
for the geometry of the cylinder, which is the natural domain in which
the Wigner function should be defined.

We have now all the ingredients needed to accomplish our program. In
what follows, the Fourier transform of $2 \pi$-periodic functions
(i.e., with domain in $\mathcal{S}_{1}$), defined as
\begin{equation}
  \label{eq:Lemma1a}
  (\mathcal{F} g)(k) =
  \frac{1}{2\pi} \intcirc  g(\phi) \,e^{ i\phi k } \, d\phi \, ,
\end{equation}
with $k \in \mathbb{Z}$, will play a relevant role. We first state 
our main result, which can be viewed as analogous to the Hudson theorem 
for the canonical pair angle and angular momentum.

\begin{theorem}[Classical OAM states]
The Wigner function of a pure state $|\psi\rangle$ is non-negative if and only 
if $|\psi\rangle$ is an OAM eigenstate $|\ell_0\rangle$.
\end{theorem}
\begin{proof} 
  The sufficiency is obvious since the Wigner function for the state
  $|\ell_0\rangle$ is $W_{| \ell_{0} \rangle} (\ell,\phi) =
  \delta_{\ell \ell_0}/(2\pi)$.  The delicate point is to prove the
  necessity. Before proceeding, we sketch the idea behind the proof.
  The first step is to show that the wave function [and thus, the
  integrand in Eq.~\eqref{eq:WignerAngle}] must be of constant
  modulus. The second step is then to corroborate that the Wigner
  function can only be non-zero for a single value of
  $\ell$. Traciality permits us to derive an equation that shows that
  this value of $\ell$ cannot vary over $\phi$, and that indeed the
  only states with non-negative Wigner functions are the OAM
  eigenstates.  We start with the following lemma.
\begin{lemma}
  If the Fourier transform of a smooth, complex, $2\pi$-periodic
  function $g(\phi)$ is non-negative, then the integration kernel
  $g(\phi-\phi^{\prime})$ is non-negative.
\end{lemma}
\begin{proof}

  By a direct calculation we can check that
  \begin{equation}
    \intcirc  g (\phi - \phi^{\prime}) \,e^{-i\phi^{\prime}k}  \, 
    d\phi^{\prime}
    =  2 \pi \, (\mathcal{F} g) (k) \, e^{-i\phi k} \, ,
  \end{equation}
  so, for any smooth test function $\chi(\phi) = \suma{k} \chi(k) \,
  e^{- i\phi k}$, it holds
  \begin{equation}
    \label{eq:Lemma1b}
    \intcirc \chi^{\ast}(\phi) \, g(\phi - \phi^{\prime}) \,
    \chi(\phi^{\prime}) \, d\phi d\phi^{\prime} = 
    4\pi^2 \suma{k} |\chi(k)|^2 \, (\mathcal{F} g)(k) \, .
  \end{equation}
  It is clear that the non-negativity of the kernel  $g(\phi - \phi^{\prime})$
  follows from the non-negativity of the Fourier transform
  $(\mathcal{F}g)(k)$.
\end{proof}

We apply the lemma to
\begin{equation}
  \label{eq:TestFun}
  \chi (\phi) = \frac{1}{2} [ \delta_{2\pi} (\phi - c_1) +
  \delta_{2\pi} (\phi - c_2)  ] \, ,
\end{equation} 
Here, $\delta_{2 \pi}$ denotes the periodic delta function (or Dirac
comb) of period $2 \pi$ and $c_{1}, c_{2} \in \mathcal{S}_{1}$. For
this function we have $|\chi (k)|^2 = \{ 1+ \cos [k(c_1 - c_2)]
\}/(8\pi^2)$, so the sum in the right-hand side of
Eq.~(\ref{eq:Lemma1b}) reduces to
\begin{equation}
  g(0)/2 + [ g(c_1 - c_2) + g(c_2 - c_1) ] /4 \, .
\end{equation}
Consequently, for a function $g(\phi )$ whose Fourier transform is
non-negative, the kernel $g(\phi - \phi^{\prime})$ must also be
non-negative on the test functions (\ref{eq:TestFun}) for all the
possible parameters $c_1,c_2 \in \mathcal{S}_1$.

For a pure state $| \psi \rangle$, the Wigner function (\ref{eq:WignerAngle}) 
is just the Fourier transform of  $\psi^{\ast}(\phi+\phi^{\prime}/2) \, 
\psi(\phi-\phi^{\prime}/2)$, where we have expressed the wave functions 
in the angle representation. By Lemma 1, for the test functions
  (\ref{eq:TestFun})  the non-negativity of $W_{| \psi \rangle}$ leads to
\begin{equation}
  \label{eq:Key}
  |\psi(\phi)|^2 \geq | \psi(\phi - a/2)| \, |\psi(\phi + a/2) | \, , 
\end{equation}
with $a = c_1 - c_2$. This implies that $|\psi(\phi)|$ cannot have any
minima and the modulus of $\psi$ must thus be flat over
$\mathcal{S}_1$.

To proceed further we need a technical detail.
\begin{lemma}
  If a function $f(k): \mathbb{Z} \to \mathbb{C}$ has an inverse
  Fourier transform of constant modulus over $\phi$, then
  \begin{equation}
    \label{eq:Lemma2b}
    \suma{k} f(k) \, f^{\ast} (k+j) = 0 \qquad \forall j \neq 0 \,.
  \end{equation}
\end{lemma}
\begin{proof}
  Let us first introduce the operator
  \begin{equation}
    \label{eq:Lemma2Proof1}
    \op{A} = \suma{m,k} f(m-k) \, |m \rangle \langle k| \, .
  \end{equation}
  One can check that it can be expressed in a diagonal form in the
  angle basis, namely
  \begin{equation}
    \label{eq:Lemma2Proof2}
    \op{A} = \intcirc | \phi \rangle \langle \phi| \,
    (\mathcal{F}^{-1} f) (-\phi) \, d\phi \, .
  \end{equation}

  If $(\mathcal{F}^{-1} f) (\phi)$ has constant modulus, it
  can be written as  $(\mathcal{F}^{-1} f) (\phi) = c \,
    e^{i\lambda (\phi)}$, 
  where $\lambda$ is a real function.  Therefore, we have $ \op{A} \,
  \op{A}^\dag =  |c|^2 \, \op{\openone} $. 
  But according to the definition~(\ref{eq:Lemma2Proof1}), this is
  tantamount to the orthogonality relation
  \begin{equation}
    \nonumber
    \suma{m, k} \suma{m^{\prime}, k^{\prime}} \langle n | m \rangle 
    \langle k| f(m - k) |k^{\prime} \rangle \langle m^{\prime} |
    f^{\ast} (m^{\prime} - k^{\prime}) |n + j \rangle  = 0 \, .
  \end{equation}
  The Plancherel formula allows one to cancel the diagonal parts, so we
  are led to
  \begin{equation}
    \label{eq:Lemma2Proof6}
    \suma{k} f(n-k) \,f^{\ast}(n+j-k)  = 0 \, ,
  \end{equation}
  whence the result follows.
\end{proof}

Next, for every $\phi$, we consider the Wigner function of the state
as a function exclusively of the discrete index $\ell$; that is,
$f_\phi(\ell) = W_{| \psi \rangle} (\ell,\phi): \mathbb{Z} \to
\mathbb{R}$ (in fact, $W$ is real valued), and make use of the fact
that the (inverse) Fourier transform of $f_\phi(\ell)$ has a constant
modulus over $\phi$.  Then, by Lemma 2, the orthogonality
\begin{equation}
  \label{eq:WignerOrthogonE2}
  \suma{\ell} f_\phi(\ell) \,f^{\ast}_\phi (\ell + \ell^{\prime}) = 0 \, , 
  \qquad \forall \ell^{\prime} \neq 0 \, ,
\end{equation}
must hold for all $\phi \in \mathcal{S}_1$. But since $f$ is
non-negative on the whole phase-space, this is only possible if $f$ is
equal to zero for all but one $\ell_0$. Note that, in principle,
$\ell_0$ may depend on $\phi$.  Taking into account the marginal
distribution \eqref{eq:ELMargin}, we see that $W (\ell,\phi) =
\delta_{\ell \ell_0 (\phi)} /(2\pi)$.

We now make use of the fact that the state $|\psi \rangle$ is pure
[that is, $\Tr (\op{\varrho}^{2}) = 1$]. From the traciality property,
one can show that the Wigner function representing the product of two
density operators $\op{\varrho}$ and $\op{\sigma}$ can be expressed as
\begin{eqnarray}
  \label{eq:Starprod2}
  &\displaystyle   W_{\op{\varrho} \, \op{\sigma}} (\ell,\phi) = 
  \frac{1}{2\pi} \suma{\ell_1 \,\ell_2} \intcirc 
  \,W_{\op{\varrho}} (\ell +\ell_1,
  \phi + \psi_1/2) &  \nonumber \\
  & \times \,
  W_{\hat\sigma}(\ell+\ell_2,\phi+\psi_2/2) \, e^{i(\ell_2
    \psi_1 - \ell_1 \psi_2)}  \, d\psi_1 d\psi_2 \, . &
\end{eqnarray}
We apply this to the pure state $|\psi\rangle$ whose Wigner function
is of the form $\delta_{\ell \ell_0 (\phi)}/(2\pi)$.
  
Without loss of generality, we can assume that $\ell_0 (\phi =0) = 0$
and may revert this choice later by a displacement $|\psi\rangle \to
\hat D( \ell_0,0)|\psi\rangle$. Then, Eq.~\eqref{eq:Starprod2} becomes
\begin{eqnarray}
  \label{eq:StarCond}
  & \displaystyle
  W_{|\psi\rangle} (0,0) = \frac{1}{2\pi} & \nonumber \\ 
  & \displaystyle
  = \frac{1}{(2\pi)^3} \intcirc   
  e^{i[\ell_0(\psi_2/2) \psi_1 -\ell_0(\psi_1/2)\psi_2 ]} \, d\psi_1 d\psi_2 \, . &
\end{eqnarray}
This means that the integral of the imaginary part must vanish, while
the integral of the real part must be equal $(2\pi)^2$.  This is only
possible if the exponential is exactly one for all the arguments
$(\psi_1,\psi_2)$; i.e., $ \ell_0 (\psi_1 /2) \,\psi_2 = \ell_0
(\psi_2/2)\,\psi_1 \, \bmod{2\pi}$.  This is only possible when
$\ell_0 \equiv 0 $.
\end{proof}

We have shown that if the Wigner function of a pure state is
non-negative, then it is necessarily a Kronecker delta, and thus
stems from an OAM eigenstate, which concludes the long yet instructive
proof of our theorem.

It is worth stressing that for the continuous case the notions of
coherent states, Gaussian wave packets, and states with non-negative
Wigner functions (often identified as nonclassical states) are
completely equivalent. However, special care must be paid in extending
these ideas to other physical systems like OAM, since they lose their
equivalence.

For example, OAM coherent states $|\ell_{0}, \phi_{0} \rangle$ in the
cylinder~\cite{Kowalski:1996} can be expressed in the angle
representation by
\begin{equation}
  \langle \phi| \ell_0 , \phi_0 \rangle  = 
  \frac{e^{i \ell_0 (\phi - \phi_0)}}
  {\sqrt{\vartheta_3 \left ( 0  \big | \frac{1}{e} \right )}}
  \vartheta_3\left(\frac{\phi-\phi_{0}}{2} \Big | \frac{1}{e^2} \right) \, ,
  \nonumber
\end{equation}
where $\vartheta_3$ denotes the third Jacobi theta function.  However,
despite the key role played by this function in angular problems, a
simple calculation~\cite{Rigas:2008ac} immediately reveals that the
Wigner function for them takes negative values.

In the same vein, the states
\begin{equation}
  \label{centered_mises}
  \Psi_{\kappa} ( \phi ) = \frac{1}{\sqrt{2 \pi I_0(2 \kappa)}} 
  \exp (\kappa \cos\phi ) \, , 
\end{equation}
whose associated probability distribution is precisely the von Mises
distribution~\cite{Rehacek:2008ss}, are usually taken as Gaussians
for this problem. One can easily check that their Wigner function 
also takes negative values.  

Even with all these cautions, the characterization we have presented
of OAM eigenstates as the only ones with non-negative Wigner function
has interest in its own, although, unfortunately, they cannot be
viewed as Gaussian states.
 
A topic of interest is the characterization of unitaries that
preserve the non-negativity. Obviously, all the displacement operators
are of this kind. But the exponential of an arbitrary real
function $ f(\op{L})$ also preserves non-negativity and this includes
quadratic exponentials, which are essential for a full quantum
reconstruction of vortex states~\cite{Rigas:2008ac}.

Finally, let us mention that a question that naturally arises is
whether our result can be extended to mixed states.  Although this
question has been approached by using the notion of the Wigner
spectrum~\cite{Brocker:1995mw} and explored quite recently for
continuous variables~\cite{Mandilara:2009vy}, in our case a simple
extension seems difficult and will be the object of our future work.

We acknowledge discussions with Hubert de Guise, Jos\'e Gracia-Bond{\'i}a,
and Hans Kastrup. This work was supported by the Spanish Research
Directorate, Grants FIS2005-06714 and FIS-2008-04356, the Mexican
Consejo Nacional de Ciencias y Tecnolog\'{\i}a (CONACyT), Grant45704,
and the Czech Ministry of Education, Projects MSM6198959213 and
LC06007.


\end{document}